\documentclass[12pt,onecolumn,draftcls]{IEEEtran}
\usepackage{amsfonts}
\usepackage{cite}
\usepackage{citesort}
\usepackage{graphicx}
\usepackage{amsmath}
\usepackage{times}
\usepackage{subfigure}
\usepackage{latexsym}
\usepackage{graphicx}
\usepackage{bm}
\usepackage{amssymb}
\usepackage[center]{caption2}
\usepackage{stfloats}
\usepackage{cases}
\usepackage{array}
\usepackage{setspace}
\usepackage{citesort}

\newtheorem{theorem}{Theorem}

\newtheorem{lemma}{Lemma}

\newcaptionstyle{mystyle2}{%
\captionlabel $.$ \, \doublespacing \captiontext \par}
\captionstyle{mystyle2}

\begin{document}
\title{Quantized CSI-Based Tomlinson-Harashima Precoding in Multiuser MIMO Systems}
\author{Liang Sun, \emph{Member}, \emph{IEEE} and Ming Lei}

\maketitle


\begin{abstract}
This paper considers the implementation of Tomlinson-Harashima (TH)
precoding for multiuser MIMO systems based on quantized channel
state information (CSI) at the transmitter side. Compared with the
results in \cite{Windpassinger04}, our scheme applies to more
general system setting where the number of users in the system can
be less than or equal to the number of transmit antennas. We also
study the achievable average sum rate of the proposed quantized
CSI-based TH precoding scheme. The expressions of the upper bounds
on both the average sum rate of the systems with quantized CSI and
the mean loss in average sum rate due to CSI quantization are
derived. We also present some numerical results. The results show
that the nonlinear TH precoding can achieve much better performance
than that of linear zero-forcing precoding for both perfect CSI and
quantized CSI cases. In addition, our derived upper bound on the
mean rate loss for TH precoding converges to the true rate loss
faster than that of zero-forcing precoding obtained in
\cite{Jindal06} as the number of feedback bits becomes large. Both
the analytical and numerical results show that nonlinear precoding
suffers from imperfect CSI more than linear precoding does.
\end{abstract}

\begin{keywords}
Tomlinson-Harashima precoding, QR decomposition, random vector
quantization, zero-forcing, Givens transformation.
\end{keywords}

%
%

\section{Introduction}
Since the pioneering work \cite{teletar99} and \cite{Foschini96},
multiple-input multiple-output (MIMO) communication systems have
been extensively studied in both academic and industry communities
and becomes the key technology of most emerging wireless standards.
It is shown that significantly enhanced spectral efficiency and link
reliability can be achieved compared with conventional single
antenna systems \cite{teletar99,Wolniansky98}. In the downlink
multiuser MIMO systems, multiple users can be simultaneously served
by exploiting the spatial multiplexing capability of multiple
transmit antennas, rather than trying to maximize the capacity of a
single-user link.

The performance of a MIMO system with spatial multiplexing is
severely impaired by the multi-stream interference due to the
simultaneous transmission of parallel data streams. To reduce the
interference between the parallel data streams, both the processing
of the data streams at the transmitter (precoding) and the
processing of the received signals (equalization) can be used.
Precoding matches the transmission to the channel. Accordingly,
linear precoding schemes with low complexity are based on
zero-forcing (ZF) \cite{Haustein02} or minimum mean-square-error
(MMSE) criteria \cite{Joham02} and their improved version of channel
regularization \cite{Peel05}. In spite of very low complexity, the
linear schemes suffer from capacity loss. Nonlinear processing at
either the transmitter or the receiver provides an alternative
approach that offers the potential for performance improvements over
the linear approaches. This kind of approaches includes schemes
employing linear precoding combined with decision feedback
equalization (DFE) \cite{Yang94,Wolniansky98}, vector
perturbation\cite{Hochwald05}, Tomlinson-Harashima (TH) precoding
\cite{Windpassinger04,Fischer02}, and ideal dirty paper coding
\cite{Costa83,Weingarten06} which is too complex to be implemented
in practice.
Vector perturbation has been proposed for multiuser MIMO channel
model and can achieve rate near capacity \cite{Hochwald05}. It has
superior performance to linear precoding techniques, such as
zero-forcing beamforming and channel inversion, as well as TH
precoding \cite{Hochwald05}. However, this method requires the joint
selection of a vector perturbation of the signal to be transmitted
to all the receivers, which is a multi-dimensional integer-lattice
least-squares problem. The optimal solution with an exhaustive
search over all possible integers in the lattice is complexity
prohibited. Although some sub-optimal solutions, such as sphere
encoder\cite{Damen00}, exist, the complexity is still much higher
than TH precoding.

TH precoding can be viewed as a simplified version of vector
perturbation by sequential generation of the integer offset vector
instead of joint selection. This technique employs modulo arithmetic
and has a complexity comparable to that of linear precoders. It was
originally proposed to combat inter-symbol interference in highly
dispersive channels\cite{Harashima72} and can readily be extended to
MIMO channels \cite{Amico08,Windpassinger04}. Although it was shown
in \cite{Hochwald05} that TH precoding does not perform nearly as
well as vector perturbation for general SNR regime, it can achieve
significantly better performance than the linear pre-processing
algorithm, since it limits the transmitted power increase while
pre-eliminating the inter-stream interference\cite{Fischer02}. Thus,
it provides a good choice of tradeoff between performance and
complexity and has recently received much
attention\cite{Windpassinger04,Fischer02}. Note that TH precoding is
strongly related to dirty paper coding. In fact, it is a suboptimal
implementation of dirty paper coding proposed in \cite{Caire03}.

As many precoding schemes, the major problem for systems with TH
precoding is the availability of the channel state information (CSI)
at the transmitter. In time division duplex systems, since the
channel can be assumed to be reciprocal, the CSI can be easily
obtained from the channel estimation during reception. In frequency
division duplex (FDD) systems, the transmitter cannot estimate this
information and the CSI has to be communicated from the receivers to
the transmitter via a feedback channel. In this paper, we will focus
on the implementation of TH precoding in FDD systems. In this
context, for linear precoding, there have been extensive research
results for MIMO systems with quantized CSI at the transmitter
\cite{Au-Yeung07,Yoo_limited,Jindal06}. However, as far as we know,
there has been very few works directed at the design of TH precoding
based on the quantized CSI at the transmitter side. 
In this respect, the previous design in \cite{Israa11} is based on
MMSE criteria. Since the MSE is a function of both statistics
(moments) of the channels and the statistics of the channel
quantization error, the computation of the MSE requires the exact
distribution of the channels which can be very difficult to obtain
in practical systems. In addition, even if the exact distribution
function of channels could be obtained, the statistics of
quantization error can be very difficult to obtain for more general
channel fading other than uncorrelated Rayleigh fading even with
simple random vector quantization (RVQ) codebook. Instead, we aim to
design low complexity method which can be easily implemented in
practical systems with arbitrary channel fading. Our scheme employs
a more direct method which only depends on the quantized CDI of user
channels.

In this paper, we design a multiuser spatial TH precoding based on
quantized CSI and ZF criteria. As in \cite{Windpassinger04}, we
focus on high spectral efficiency, in particular non-binary
modulation alphabets and correspondingly we assume high
signal-to-noise ratios (SNRs). In contrast to \cite{Windpassinger04}
where perfect CSI is at the transmitter side, we assume only
quantized CDI is available at the transmitter. The feedforward
filter as well as the feedback filter are computed at the
transmitter only based on the available quantized CDI at the
transmitter side. In addition, our scheme also generalizes the
results in \cite{Windpassinger04} to more general system setting
where the number of users $K$ in system can be less than or equal to
the number of transmit antennas $n_T$. We also study the achievable
average sum rate of the proposed quantized CSI-based TH precoding
scheme by analytically characterizing the average sum rate and the
rate loss due to quantized CSI as functions of the number of
feedback bits per user. Our derived upper bound for TH precoding
tracks the true rate loss quite closely and appears to converge
faster than the upper bound for ZF precoding obtained in
\cite{Jindal06} as the number of feedback bits becomes large.


\begin{figure}[t]
\centering
\includegraphics[width= 0.8\columnwidth]{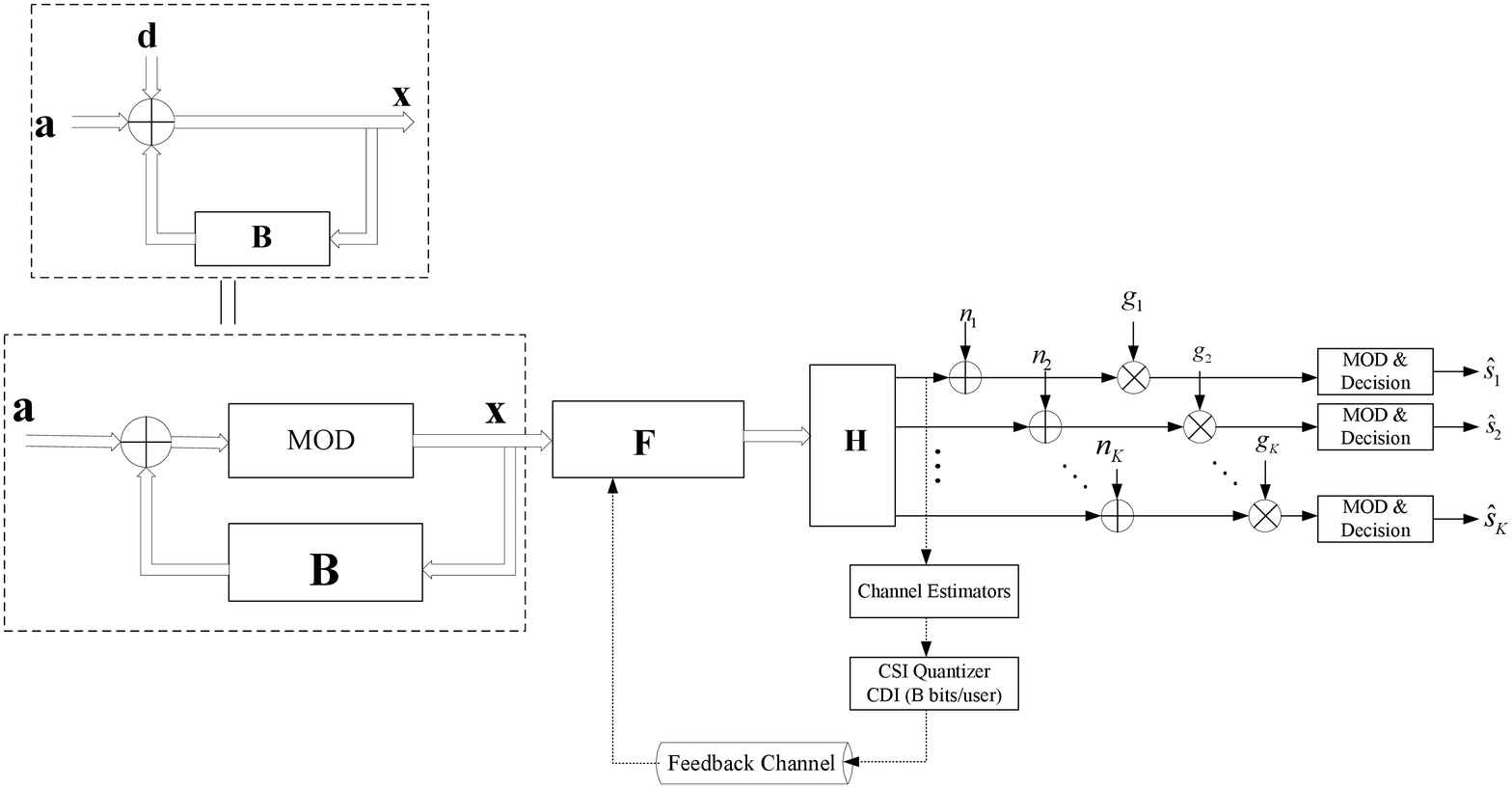}
\caption{TH precoding for multiuser MIMO downlink with quantized CSI
feedback.} \label{fig:system}
\end{figure}

\section{System Model}\label{sec:model_perfect_CSI}
As shown in Fig. \ref{fig:system}, we consider the multi-user
downlink systems where TH precoding \cite{Harashima72} is used at
the transmitter for multi-user interference pre-subtraction. The
transmitter is equipped with $n_T$ transmit antennas and $K$
decentralized users each has a single antenna such that $K \leq
n_T$. Let the vector $\mathbf{s} = [s_1, \cdots, s_K]^T \in
\mathbb{C}^{K\times 1}$ represent the modulated signal vector for
all users, where $s_k$ is the $k$-th modulated symbol stream for
user $k$. Here we assume that an $M$-ary square constellation ($M$
is a square number) is employed in each of the parallel data streams
and the constellation set is $\mathcal{A}= \big\{ s_I +j s_Q
|~s_I,s_Q \in \pm 1 \sqrt{\frac{3}{2(M-1)}}, \pm 3 \sqrt{\frac{3}{2
(M-1)}}, \cdots, \pm \sqrt{M}\sqrt{\frac{3}{2(M-1)}}\big\}$. In
general, the average transmit symbol energy is normalized, i.e.
$\mathbb{E}\{|s_k|^2\} = 1$. $\mathbf{s}$ is fed to the precoding
unit, which consists of a backward square matrix $\mathbf{B}$ and a
nonlinear operator $\text{MOD}_{\tau}(\cdot)$ which acts
independently over the real and imaginary parts of its input as
follows
\begin{equation}\label{eq:modM}
\text{MOD}_{\tau} (x) = x -\tau \bigg\lfloor \frac{x+\tau}{2\tau}
\bigg\rfloor ,
\end{equation}
where $\tau = \sqrt{M}\sqrt{\frac{3}{2(M-1)}}$, $\lfloor z \rfloor$
is the largest integer not exceeding $z$. $\mathbf{B}$ must be
strictly lower triangular to allow data precoding in a recursive
fashion \cite{Windpassinger04}. The construction of $\mathbf{B}$
will depend on the level of CSI of the supported users available at
the transmitter side. If we temporarily neglect the nonlinear
operator $\text{MOD}_{\tau} (\cdot)$ in Fig. \ref{fig:system}, the
channel signal vector $\mathbf{x} = [x_1, \cdots, x_K]^T$ can be
generated as
\begin{align}\label{eq:THP1}
x_1 &= s_1, \nonumber\\
x_k &= s_k - \sum_{l=1}^{k-1} [\mathbf{B}]_{k,l}x_l, ~ k =2, \cdots,
K.
\end{align}
In this way, if $\mathbf{b} = \left( \mathbf{I} + \mathbf{B} \right
)^{-1} \mathbf{s}$ become large in the presence of deep fading the
transmit power can be greatly increased. TH precoding modulo
(\ref{eq:modM}) reduces the transmit symbols into the boundary
square region of $ \mathcal{R} = \{x+jy | x, y\in (-\tau, \tau)\}$.
With (\ref{eq:modM}) and (\ref{eq:THP1}) the channel signals are
equivalently given as
\begin{equation}
x_k =  s_k + d_k - \sum_{l=1}^{k-1} [\mathbf{B}]_{k,l}x_l, ~ k =2,
\cdots, K,
\end{equation}
where $d_k \in \left\{ 2 \tau (p_I + j p_Q) | ~ p_I , p_Q \in
\mathbb{Z} \right\}$ is properly selected to ensure the real and
imaginary parts of $x_k$ are constrained into $\mathcal{R}$
\cite{Windpassinger04}. The constellation of the modified data
symbols $ v_k = s_k + d_k $ is simply the periodic extension of the
original constellation along the real and imaginary axes.
Equivalently, the effective data symbols $ v_k $ ($k =2, \cdots, K$)
are passed into $\mathbf{B}$, which is implemented by the feedback
structure. Thus, we have
\begin{equation}
\mathbf{v} = \mathbf{C} \mathbf{x},
\end{equation}
where $\mathbf{v} = [v_1,v_2, \cdots, v_K]^T$ and
\begin{equation}
\mathbf{C} = \mathbf{B} +\mathbf{I}.
\end{equation}
We will make the standard observation that the elements of
$\mathbf{x}$ are almost uncorrelated and uniformly distributed over
the Voronoi region of the constellation $\mathcal{R}$, and that such
a model becomes more precise as $n_T$ increases \cite[Theorem
3.1]{Fischer02}. With $\mathbb{E}\left\{ \mathbf{s}\mathbf{s}^H
\right\} = \mathbf{I}$, the covariance of $\mathbf{x}$ can be
accurately approximated as $\mathbf{R_x} =
\frac{M}{M-1}\mathbf{I}$\cite{Fischer02}. Moreover, the induced
shaping loss by the non-Gaussian signaling leads to the fact that
the achievable rate can be up to $1.53$ dB from the channel capacity
\cite{Wesel98}. However, as indicated in \cite{Windpassinger04}, the
so-called shaping loss can be bridged by higher-dimensional
precoding lattices. A scheme named ``inflated lattice'' precoding
has been proved to be capacity-achieving in \cite{Uri04}. Thus,
following \cite{Windpassinger04}, we will ignore the shaping gap in
this work.

A spatial channel pre-equalization is performed at the transmitter
side using a feedforward precoding matrix $\mathbf{F} \in
\mathbb{C}^{n_T\times K} $. Throughout this work, we assume equal
power allocation to all supported users. Then the received signal
can be written as
\begin{equation}
\mathbf{r} = \sqrt{\frac{P}{\kappa}}\mathbf{H}\mathbf{F}\mathbf{x} +
\mathbf{n},
\end{equation}
where $\mathbf{H} = \left[ \mathbf{h}_1^T, \cdots, \mathbf{h}_K^T
\right]^T $ is the compact flat fading channel matrix consisting of
all users's channel vectors and $\mathbf{h}_k \in \mathbb{C}^{n_T
\times 1} $ is the channel from the transmitter to user
$k$\footnote{The ordering of the users' channel vectors in
$\mathbf{H}$ will affect the precoding order of the users'
information signals and further affect the performance of each user.
However, at this stage we assume the user channel vectors are
randomly ordered. Thus the TH precoding order of the users is $1, 2,
\cdots, K$.}. $\kappa$ is used for transmit power normalization.
$\mathbf{F} \in \mathbb{C}^{n_T \times K}$ satisfies transmit power
constraint $\frac{P}{\kappa}\text{Tr}
\{\mathbf{F}\mathbf{R}_{\mathbf{x}}\mathbf{F}^H\}
=\frac{P}{\kappa}\frac{M}{M-1}\text{Tr}\{\mathbf{F}\mathbf{F}^H\}=
P$. As for $\mathbf{B}$, $\mathbf{F}$ is also designed based on the
level of CSI available to the transmitter. We assume $\mathbf{n}$ is
the white additive noise at all the receivers with the covariance
$\mathbf{R}_{\mathbf{n}} = \mathbf{I}$ without loss of generality.
Each receiver compensates for the channel gain by dividing by a
factor $g_k$ prior to the modulo operation as follows:
\begin{align}
\mathbf{y} & = \mathbf{G} \left ( \sqrt{\frac{P}{\kappa}}
\mathbf{H}\mathbf{F}\mathbf{x} + \mathbf{n} \right ),
\end{align}
where $\mathbf{G} = \text{diag} \left( g_{1,1}, \cdots, g_{K,K}
\right)$.

Throughout this work, we assume each receiver can obtain perfect CSI
of his own through channel estimation and feeds back this
information to the transmitter via a zero-delay feedback link with
possible rate-constraints. In this part, we assume this feedback
information is perfect at the transmitter. The design of TH
precoding for MU MIMO systems with perfect CSI has been studied in
\cite{Windpassinger04} where, for simplicity, the authors restrict
the work to the systems with equal number of transmit antennas and
users. In this part, we review and extend that construction for
systems with arbitrary number of users no more that of the transmit
antennas.

With perfect CSI at the transmitter, let the QR decomposition of the
compact channels be $\mathbf{H} = \mathbf{RQ}$, where $\mathbf{R} =
\left[r_{i,j}\right] \in \mathbb{C}^{K \times K}$ is a lower left
triangular matrix and $\mathbf{Q} \in \mathbb{C}^{K \times n_T}$ is
a semi-unitary matrix with orthonormal rows which satisfies
$\mathbf{Q} \mathbf{Q}^H = \mathbf{I}$. Then the precoding matrix
$\mathbf{F}$ is given as $\mathbf{F}= \mathbf{Q}^H$, the scaling
matrix $\mathbf{G}$ is given as $\mathbf{G} =
\sqrt{\frac{\kappa}{P}}\mathbf{\Delta}$ with $\mathbf{\Delta}=
\text{diag} \left( r_{1,1}^{-1}, \cdots, r_{K,K}^{-1} \right)$ and
the feedback matrix reads $ \mathbf{B} = \mathbf{\Delta HF}
-\mathbf{I} = \mathbf{\Delta R} -\mathbf{I}$. According to the
transmit power constraint $\frac{P}{\kappa} \frac{M}{M-1} K = P$, we
have $\kappa = \frac{M}{M-1} K$. With the processing, the effective
received data symbols $\mathbf{y}$ corrupted by additive noise can
be written as\cite{Windpassinger04}
\begin{equation}
\mathbf{y} = \mathbf{v} + \mathbf{G}\mathbf{n}.
\end{equation}
At the receivers, each symbol in $\mathbf{y}$ is firstly modulo
reduced into the boundary region of the signal constellation
$\mathcal{A}$. A quantizer of the original constellation will follow
the modulo operation to detect the received signals. The SNR $\xi_k$
for receiver $k$ can be written as
\begin{equation}\label{eq:SNR_per}
\xi_k = \frac{P}{\kappa} |r_{k,k}|^2.
\end{equation}

In the following part, we will describe how to implement the
precoding with quantized CSI obtained at the transmitter.

\section{System with Quantized transmit CSI}
In practical systems, perfect CSI is never available at the
transmitter. For example, in a FDD system, the transmitter obtains
CSI for the downlink through the limited feedback of $B$ bits by
each receiver. Following the studies of quantized CSI feedback in
\cite{Jindal06,Yoo_limited}, channel direction vector is quantized
at each receiver, and the corresponding index is fed back to the
transmitter via an error and delay-free feedback channel. Given the
quantization codebook $\mathbb{W} = \{\mathbf{w}_1,\cdots,
\mathbf{w}_{n}\}$ ($\mathbf{w}_{i} \in \mathbb{C}^{1\times n_T} $),
which is known to both the transmitter and all the receivers, the
$k$-th receiver selects the quantized channel direction vector of
its own channel as follows:
\begin{equation}
{\hat{\mathbf{h}}_k} = \text{arg} \max_{\mathbf{w}_i
\in\mathbb{W}}\{ |\bar{{\mathbf{h}}}_k \mathbf{w}_i|^2\},
\end{equation}
where $\bar{{\mathbf{h}}}_k = \frac{\mathbf{h}_k}{
\|\mathbf{h}_k\|}$ is the channel direction vector of user $k$.

In this work, we use RVQ codebook, in which the $n$ quantization
vectors are independently and isotropically distributed on the
$n_T$--dimensional complex unit sphere. Although RVQ is suboptimal
for a finite-size system, it is very amenable to analysis and also
its performance is close to the optimal quantization\cite{Jindal06}.
Using the result in \cite{Jindal06}, for user $k$ we have
\begin{align}\label{eq:channel_decom}
\bar{\mathbf{h}}_k = \hat{\mathbf{h}}_k\cos\theta_k +
\tilde{\mathbf{h}}_k\sin\theta_k,
\end{align}
where $\cos^2\theta_k = |\bar{\mathbf{h}}_k\hat{\mathbf{h}}_k^H|^2$,
$\tilde{\mathbf{h}}_k \in \mathbb{C}^{ 1\times n_T}$ is a unit norm
vector isotropically distributed in the orthogonal complement
subspace of $\hat{\mathbf{h}}_k$ and independent of $\sin\theta_k$.
Then $\mathbf{H}$ can be written as
\begin{equation}
\mathbf{H} = \mathbf{\Gamma} \left( \mathbf{\Phi} \hat{\mathbf{H}} +
\mathbf{\Omega} \tilde{\mathbf{H}} \right),
\end{equation}
where $\mathbf{\Gamma} = \text{diag} \big(\rho_1,\cdots ,\rho_K
\big) $ with $ \rho_k = \|\mathbf{h}_k\| $, $\mathbf{\Phi}=
\text{diag} \left(\cos\theta_1,\cdots ,\cos\theta_K \right)$ and
$\mathbf{\Omega} =\text{diag} \big(\sin\theta_1,\\ \cdots ,
\sin\theta_K\big) $, $\hat{\mathbf{H}} = \left
[\hat{\mathbf{h}}_1^T, \cdots ,\hat{\mathbf{h}}_K^T \right ]^T $ and
$\tilde{\mathbf{H}} = \left [\tilde{\mathbf{h}}_1^T, \cdots
,\tilde{\mathbf{h}}_K^T \right ]^T $. For simplicity of analysis, in
this work we consider the quantization cell approximation used in
\cite{Mukkavilli03,Yoo_limited}, where each quantization cell is
assumed to be a Voronoi region of a spherical cap with surface area
approximately equal to $\frac{1}{n}$ of the total surface area of
the $n_T$-dimensional unit sphere. For a given codebook
$\mathbb{W}$, the actual quantization cell for vector
$\mathbf{w}_i$, $\mathcal{R}_i = \left\{\bar{\mathbf{h}}:
|\bar{\mathbf{h}} \mathbf{w}_i|^2 \geq |\bar{\mathbf{h}}
\mathbf{w}_j|^2, \forall ~ i \neq j \right\} $, is approximated as $
\tilde{\mathcal{R}}_i \approx \left\{\bar{\mathbf{h}}:
|\bar{\mathbf{h}} \mathbf{w}_i| \geq  1- \delta \right\}$, where
$\delta = 2^{- \frac{B}{n_T -1}}$.

With the quantized CDI at the transmitter side, the transmitter
obtains the feedforward precoding matrix $\mathbf{F}$ and feedback
matrix $\mathbf{B} $ through the QR decomposition of compact channel
matrix $\hat{\mathbf{H}} $ in the same way as the QR decomposition
of matrix $\mathbf{H}$, i.e. $\hat{\mathbf{H}} =
\hat{\mathbf{R}}\hat{\mathbf{Q}}$, where the matrices
$\hat{\mathbf{R}}$ and $\hat{\mathbf{Q}}$ have the same structure as
the matrices $\mathbf{R}$ and $\mathbf{Q}$ respectively. Then we
have $\mathbf{F} = \hat{\mathbf{Q}}^H $ and $\mathbf{B} =
\left(\text{diag}\left\{\hat{\mathbf{R}}
\right\}\right)^{-1}\hat{\mathbf{R}}- \mathbf{I}$. In addition, the
scaling matrix at the receivers now becomes
\begin{equation}
\mathbf{G} = \sqrt{\frac{\kappa}{P}} \left( \mathbf{\Gamma}
\mathbf{\Phi} ~\text{diag}\left\{\hat{\mathbf{R}}
\right\}\right)^{-1}.
\end{equation}
Using the same operation at the receiver side as that in perfect CSI
case to detect the received signals, the detected signal vector
$\hat{\mathbf{y}}$ can be further written as
\begin{align}\label{eq:Rx_sig_limited}
\hat{\mathbf{y}} &=  \mathbf{G}\left(
\sqrt{\frac{P}{\kappa}}\mathbf{H} \mathbf{F} \mathbf{x}
+\mathbf{n}\right)\nonumber\\
&= \mathbf{G} \sqrt{\frac{P}{\kappa}}\mathbf{\Gamma}  \left(
\mathbf{\Phi} \hat{\mathbf{H}} + \mathbf{\Omega} \tilde{\mathbf{H}}
\right) \mathbf{F} \mathbf{x} + \mathbf{G}\mathbf{n}\nonumber\\
&= \mathbf{v} + \left( \mathbf{\Phi}
~\text{diag}\left\{\hat{\mathbf{R}}
\right\}\right)^{-1}\mathbf{\Omega}
\tilde{\mathbf{H}}\hat{\mathbf{Q}}^H\mathbf{x} +
\sqrt{\frac{\kappa}{P}}\left( \mathbf{\Gamma}
\mathbf{\Phi}~\text{diag}\left\{\hat{\mathbf{R}}
\right\}\right)^{-1}\mathbf{n},
\end{align}
where we have used the relationship $\mathbf{v} =
\left(\text{diag}\left\{\hat{\mathbf{R}}
\right\}\right)^{-1}\hat{\mathbf{R}} \mathbf{x}$. In
(\ref{eq:Rx_sig_limited}), the first term is the useful signal
vector for all the users and the second term is interference signal
caused by the quantized CSI.

According to (\ref{eq:Rx_sig_limited}), the output
signal-to-interference-plus-noise ratio (SINR) $\gamma_k$ for
receiver $k$ can be written as
\begin{align}\label{eq:SINR_limited}
\gamma_k &= \frac{1}{
\frac{\sin^2\theta_k}{|\hat{r}_{k,k}|^2\cos^2\theta_k}
\|\tilde{\mathbf{h}}_k \hat{\mathbf{Q}}^H \|^2+
\frac{\kappa}{P}\frac{1}{\rho_k^2|\hat{r}_{k,k}|^2 \cos^2\theta_k} } \nonumber\\
&= \frac{\frac{P}{\kappa} \rho_k^2|\hat{r}_{k,k}|^2 \cos^2\theta_k}{
\frac{P}{\kappa} \rho_k^2 \| \tilde{\mathbf{h}}_k \hat{\mathbf{Q}}^H
\|^2 \sin^2\theta_k+1 }.
\end{align}

\section{Average Sum Rate Analysis under Quantized CSI Feedback}

In this section we will study the achievable average sum rate of the
proposed quantized CSI feedback TH precoding scheme. Although the
exact distribution of each term in the expression of the output SINR
$\gamma_k$ in $(15)$ can be obtained (see for the detailed
information), these terms are located at both the numerator and the
denominator in $(15)$. Thus, to obtain the exact closed-form
expression of the distribution of output SINR $\gamma_k$ can be very
difficult if not impossible, not to mention the exact closed-form
expression of the average sum rate. Thus, to simplify analysis, we
have appealed to studying some bounds of the average sum rate and
the average sum rate loss instead of exact results.
For tractability, throughout this section we assume each user's
channel is Rayleigh-faded. In the following subsection, we will
first study the statistical distribution of the power of
interference signal at each user caused by quantized CSI.

\subsection{Interference Part}
In this subsection, assuming Rayleigh fading channel and RVQ for
quantized CSI feedback, we will derive the statistical distribution
of interference part $\frac{P}{\kappa} \rho_k^2 \|
\tilde{\mathbf{h}}_k \hat{\mathbf{Q}}^H \|^2 \sin^2\theta_k$ in
(\ref{eq:SINR_limited}). It is well known that $ \rho_k^2 $ has a
$\chi_{2n_T}^2$ distribution and the distribution of $
\sin^2\theta_k$ is given in \cite{Au-Yeung07,Jindal06}. However,
since $\tilde{\mathbf{h}}_k \perp \hat{\mathbf{h}}_k $ ($k=1,\cdots,
K$) and $\hat{\mathbf{Q}}$ is determined by $\hat{\mathbf{h}}_k$
($k=1,\cdots, K$), $\tilde{\mathbf{h}}_k $ for $k=1,\cdots, K$ are
not independent of $\hat{\mathbf{Q}}$. The distribution of the term
$\| \tilde{\mathbf{h}}_k \hat{\mathbf{Q}}^H \|^2$ is still unknown
and to obtain the exact result is not trivial. The following lemma
presents the exact distribution of this interference term. It is one
of the key contributions of this paper.
\begin{lemma} \label{lemma:beta_k_pdf}
For $ 1 <K < n_T$, the random variables $\varepsilon_k = \|
\tilde{\mathbf{h}}_k \hat{\mathbf{Q}}^H \|^2$ for $k=1, \cdots, K$
follow the same beta distribution with shape $(K-1)$ and $(n_T-K)$
which is denoted as $\varepsilon_k \sim \mathrm{Beta}(K-1, n_T-K)$.
In addition, the probability density function (p.d.f.) of
$\varepsilon_k$ is given as
\begin{eqnarray}\label{eq:pdf_interfers}
f_{\varepsilon_k} (x) = \frac{1}{\beta(K-1, n_T-K)}x^{K-2}(1-
x)^{n_T-K-1}.
\end{eqnarray}
where $\beta(a, b) = \int_{0}^{1} t^{a-1} t^{b-1} {\rm d}t$
is beta function \cite{Gradshteyn2000}. Specially, when $K=1$ there
is no interference term. When $K=n_T$, $\varepsilon_k = \|
\tilde{\mathbf{h}}_k \hat{\mathbf{Q}}^H \|^2$ is equal to $1$ which
is a constant.
\end{lemma}
\begin{proof}
See Appendix \ref{proof:lemma_beta_k_pdf}.
\end{proof}
\emph{Lemma} \ref{lemma:beta_k_pdf} implies a very interesting
result that, with randomly ordered user channel vectors, the signal
of the user which is precoded ahead suffers from the same
interference signal power as the signals of the users which are
precoded afterwards. In the following we will only focus on the
general situation that $1 < K < n_T $. However, it is easy to check
that all the obtained results also apply to the special cases of
$K=1$ and $K = n_T$.

The expectation of the logarithm of the interference term
$\varepsilon_k$, which is shown to be useful in the following
theorems, is obtained in the following lemma.
\begin{lemma}\label{lemma:exp_interference}
The expectation of the logarithm of the interference term
$\varepsilon_k$ is given by
\begin{align}\label{eq:exp_interference}
\mathbb{E}_{\mathbf{H}, \mathbb{W}} \left [- \log_2 \left( \|
\tilde{\mathbf{h}}_k \hat{\mathbf{Q}}^H \|^2 \right) \right ] =
\log_2 e \sum_{m = K-1}^{n_T-2} \sum_{l=0}^{n_T -m-2}
\frac{(n_T-2)!}{ m!~ l!~ (n_T -m-2 - l) !} (-1)^{l} \frac{1}{m+l}
\end{align}
\end{lemma}
\begin{proof}
See Appendix \ref{sec:proof_exp_interference}.
\end{proof}

\subsection{Upper Bounds on the Average Sum Rate Loss and Sum Rate}
The instantaneous achievable rates for user $k$ with perfect CSI and
quantized CSI feedback are given as
\begin{equation}\label{eq:rate_per}
R_{P,k}= \log_2\left(1+ \xi_k \right)
\end{equation}
and
\begin{equation}\label{eq:rate_limit}
R_{Q,k} = \log_2(1 + \gamma_k),
\end{equation}
respectively. The following theorem quantifies the average sum rate
performance degradation as a function of the feedback rate.
\begin{theorem}\label{theorem:rate_loss}
With $B$ feedback bits per user, the average sum rate loss of user
$k$ due to quantized CSI feedback can be upper bounded
by\footnote{Note that, in contrast to ZF precoding, for TH precoding
different users have different average sum rate loss. Interestingly,
simulation results show that, for finite SNR, the users precoded
earlier will suffer from greater sum rate loss. However, in this
work will adopt the average sum rate loss over all supported users.}
\begin{align}\label{eq:sum_rate_loss}
\Delta R_{k} & = \mathbb{E}_{\mathbf{H}, \mathbb{W}}  \{ R_{P,k} - R_{Q,k} \}  \nonumber\\
& \leq \Delta R = \log_2\left( 1+  c P ~ 2^{-\frac{B}{n_T-1}}
\right) + \frac{\log_2(e)}{n_T-1} \sum_{i=1}^{n_T-1} \beta\left( n,
\frac{i}{n_T-1}\right),
\end{align}
where $c = \frac{(K-1)n_T }{\kappa(n_T-1)}$ and $n = 2^B$ is the
size of codebook.
\end{theorem}
\begin{proof}
See Appendix \ref{proof:theorem_rateloss}.
\end{proof}

According to the results in \cite[\emph{Theorem 1}]{Jindal06}, the
average sum rate loss due to quantized feeback for ZF precoding is
upper bounded by $\Delta R_{zf} < \log_2\left(1+ P ~
2^{-\frac{B}{n_T-1}} \right)$. We find the first term at the right
hand side (RHS) of (\ref{eq:sum_rate_loss}) can be approximated as
$\log_2\left(1+ P ~ 2^{-\frac{B}{n_T-1}} \right)$ with high order
constellation, large number of transmit antennas and large number of
supported users. Thus, the second term at the RHS of
(\ref{eq:sum_rate_loss}) can be seen as the sum rate degradation of
nonlinear precoding compared with that of linear precoding when only
quantized CSI is available at the transmitter side. In addition,
similar to the results for the linear ZF beamforming in
\cite{Jindal06}, the rate loss for nonlinear precoding is also an
increasing function of the system SNR ($P$). Thus, the system with
\emph{fixed} feedback rate is interference-limited at high SNR
regime, which is shown in the following theorem.
\begin{theorem}\label{theorem:sum_rate}
The average sum rate of user $k$ achieved by quantized CSI-based TH
precoding with $B$ feedback bits per user is bounded as
\begin{align}\label{eq:ave_sumrate_bound1}
R_{Q,k} \leq \log_2 e \left( \sum_{m = K-1}^{n_T-2} \sum_{l=0}^{n_T
-m-2} \frac{(n_T-2)!}{ m!~ l!~ (n_T -m-2 - l) !} (-1)^{l}
\frac{1}{m+l} + \frac{1}{n_T-1} \sum_{l=1}^{n} \frac{1}{l}  \right),
\end{align}
where $n = 2^B$ is the size of codebook.
\end{theorem}
\begin{proof}
See Appendix \ref{proof:theorem_sum_rate}.
\end{proof}
We can see from this theorem that, with fixed feedback bits per
user, as the interference and signal power both increase linearly
with $P$, the system becomes interference-limited and the average
sum rate converges to an upper bound. These can also be observed
from the simulation results in Fig. \ref{fig:sum_rate_SNR}.

In the context of linear ZF precoding in \cite{Jindal06}, the author
showed the interference-limited scenario can be avoided by scaling
the feedback rate linearly with the SNR $P_{\text{dB}}$ (in
decibels). Particularly, it is showed in \cite[\emph{Theorem}
3]{Jindal06} that in order to maintain a constant average sum rate
loss no greater than $\log_2 b$ bits per user between the system
with perfect CSI and the system with finite-rate feedback, it is
sufficient to scale the number of feedback bits per user according
to
\begin{align}\label{eq:scale_linear}
B = (n_T-1) \frac{\log_2 10}{10} P_{\text{dB}} - (n_T -1)
\log_2(b-1).
\end{align}
However, for nonlinear TH precoding, the explicit relationship
between the feedback rate and the SNR to maintain a constant average
sum rate loss cannot be easily obtained. This is mainly due to the
fact that the expression of average sum rate loss in
(\ref{eq:sum_rate_loss}) is a much more complex function of $n$
($B$) than the corresponding expression for linear ZF precoding
given in \cite[\emph{Theorem} 1]{Jindal06}. In the following we will
derive a corresponding relationship for the system employing TH
precoding. First, in Appendix \ref{proof} we show that the second
term at the RHS of (\ref{eq:sum_rate_loss}) can be bounded by a
decreasing function of $n$ for a fixed $n_T$. In addition, this
upper bound approaches zero as $n \rightarrow \infty$. Thus, as $n$
scales linearly with SNR (in decibels), for an arbitrary given
constant $ 0 < \varepsilon < 1$, we can always find a positive
integer $N(\varepsilon)$ such that, whenever $n \geq
N(\varepsilon)$,
\begin{align}\label{eq:scale_bound1}
\Delta R \leq \log_2\left( 1+  c P ~ 2^{-\frac{B}{n_T-1}} \right)
+\varepsilon.
\end{align}
To characterize a sufficient condition of the scaling of feedback
rate, we set the RHS of (\ref{eq:scale_bound1}) to be the maximum
allowable gap of $\log_2 b$. After some simple manipulations, we get
\begin{align}\label{eq:scale_bound}
B & = (n_T - 1 ) \log_2 P - \log_2(b -2^{\varepsilon} -1 ) + \log_2 c \nonumber\\
& = (n_T - 1 )  \frac{\log_2 10}{10} P_{\text{dB}} - \log_2(b
-2^{\varepsilon} -1 ) + \log_2 c.
\end{align}

\begin{figure}[t]
\centering
\includegraphics[width= 0.8\columnwidth]{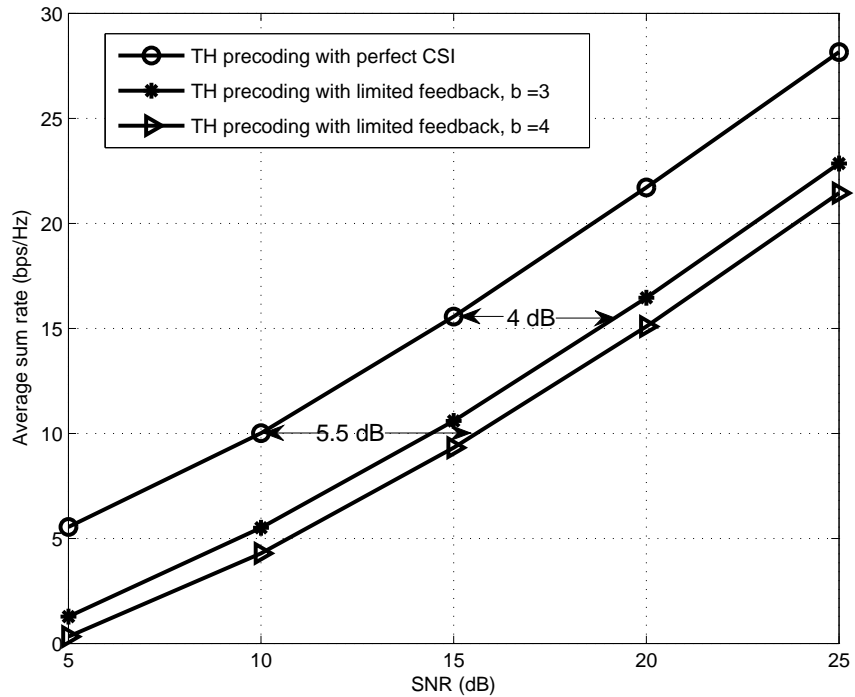}
\caption{$4 \times 4 $ system with increasing number of feedback
bits.} \label{fig:increasing_bits}
\end{figure}

In Fig. \ref{fig:increasing_bits}, the average sum rate curves are
shown for a system with $n_T =4$ and $K=4$. The feedback rate is
assumed to scale according to the relationship given in
(\ref{eq:scale_bound}). Notice that, since $\varepsilon$ can be set
to be a small number when $B$ is large enough, in the simulation we
set $\varepsilon = 0$ to get a \emph{stronger} condition than
(\ref{eq:scale_bound}). Quantized CSI-based TH precoding is seen to perform within
around $4$ dB and $5.5$ dB of TH precoding with perfect CSI for $b =3$
and $b=4$ respectively.

\begin{figure}[t]
\centering
\includegraphics[width= 0.8\columnwidth]{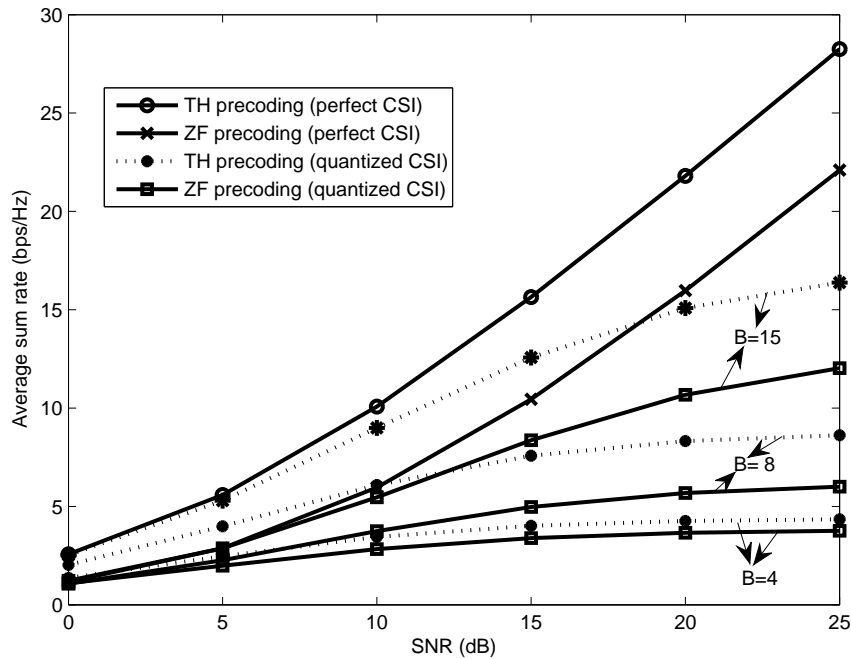}
\caption{The average sum rate performance of TH precoding and ZF
precoding for both perfect CSI and quantized CSI.}
\label{fig:sum_rate_SNR}
\end{figure}

\begin{figure}[t]
\centering
\includegraphics[width= 0.8 \columnwidth]{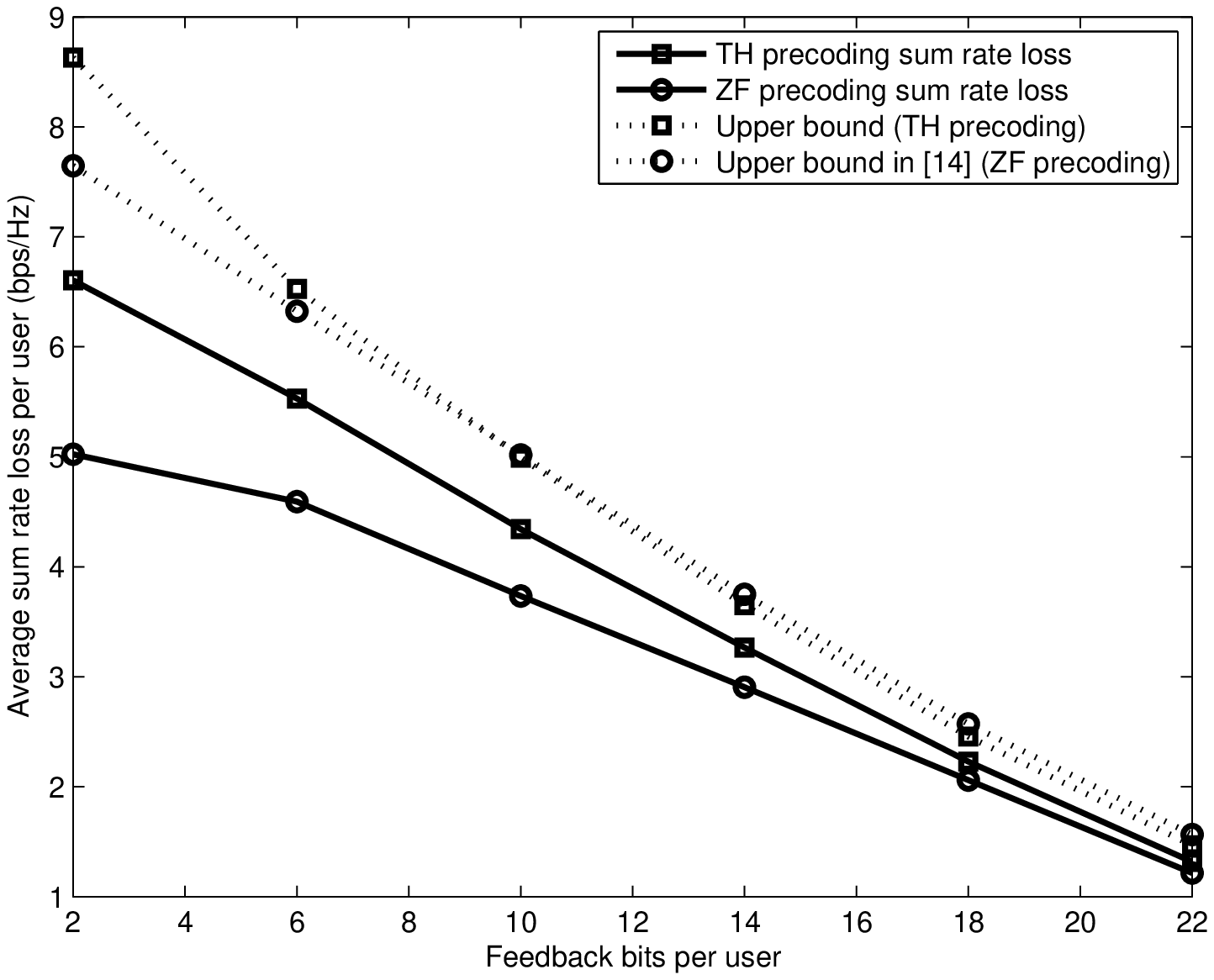}
\caption{The average sum rate loss per user and corresponding upper
bounds against the number of feedback bits. $P= 25$ dB.}
\label{fig:sum_rate_loss}
\end{figure}

\section{Numerical Results}
In this section we present some numerical results. We assume $n_T =
K =4$. Here the SNR of the systems is defined to be equal to $P$.

Fig. \ref{fig:sum_rate_SNR} shows the average sum rate performance
of TH precoding and linear ZF precoding with both perfect CSI and
quantized CSI, and $4, 8$ and $15$ feedback bits per user. We can
see TH precoding performs better than linear precoding in both
perfect CSI and quantized CSI cases. When the SNR is small and
moderate, the average sum rate achieved by quantized CSI-based TH
precoding can even be better than that of perfect CSI-based linear
ZF precoding.

Fig. \ref{fig:sum_rate_loss} plots the average sum rate loss per
user as a function of the number of feedback bits for both ZF
precoding and TH precoding in a system at an SNR of $25$ dB. We also
plot the upper bound from \emph{Theorem} \ref{theorem:rate_loss} in
this paper and the upper bound from \emph{Theorem} $1$ in
\cite{Jindal06}. From the figure we can see that nonlinear precoding
suffers from imperfect CSI more than linear precoding does. However,
the performance of nonlinear precoding can still be better than
linear precoding when SNR is not large or the feedback quantization
resolution is high enough. In addition, we notice that the upper
bound for TH precoding tracks the true rate loss quite closely, and
appears to converge faster than the upper bound for linear precoding
obtained in \cite{Jindal06} as $B$ increases.

\section{Conclusion}
In this paper, we have investigated the implementation of TH
precoding in the downlink multiuser MIMO systems with quantized CSI
at the transmitter side. In particular, our scheme generalized the
results in \cite{Windpassinger04} to more general system setting
where the number of users $K$ in the systems can be less than or
equal to the number of transmit antennas $n_T$. In addition, we
studied the achievable average sum rate of the proposed scheme by
deriving expressions of upper bounds on both the average sum rate
and the mean loss in sum rate due to CSI quantization. Our numerical
results showed that the nonlinear TH precoding could achieve much
better performance than that of linear zero-forcing precoding for
both perfect CSI and quantized CSI cases. In addition, our derived
upper bound for TH precoding converged to the true rate loss faster
than the upper bound for zero-forcing precoding obtained in
\cite{Jindal06} as the number of feedback bits increased.

\appendices

\section{Proof of \emph{Lemma} \ref{lemma:beta_k_pdf}}
\label{proof:lemma_beta_k_pdf} The results for the special cases
that $K=1$ and $n_T$ are trivial. In the following we will consider
the cases that $1 < K < n_T$. Since the user channel vectors in
$\mathbf{H} $ are unordered, so are the quantized channel direction
vectors in $\hat{\mathbf{H}} = \left [\hat{\mathbf{h}}_1^T, \cdots
,\hat{\mathbf{h}}_K^T \right ]^T $. According to the QR
decomposition of $\hat{\mathbf{H}}$ we have
\begin{equation}
\hat{\mathbf{h}}_k = \sum_{l=1}^{k} \hat{r}_{k,l}
\hat{\mathbf{q}}_l,
\end{equation}
If we require $\hat{r}_{i,i} > 0 $ for $i =1,\cdots, K$, this
decomposition is \emph{unique}. Particularly, we have $\hat{r}_{1,1}
= 1$ and $ \hat{\mathbf{q}}_1 =\hat{\mathbf{h}}_1 $. In addition,
$\tilde{\mathbf{h}}_k$ is isotropically distributed in the null
space of $\hat{\mathbf{h}}_k$\cite{Jindal06}. Thus, for $k=1$ we
have $\tilde{\mathbf{h}}_1 \bot ~ \hat{\mathbf{q}}_1$ or
equivalently $\tilde{\mathbf{h}}_1$ is an isotropically distributed
unit vector in the null space of $\hat{\mathbf{q}}_1$.

With the assumption of RVQ, the quantized channel direction vectors
$\hat{\mathbf{h}}_k ( k =1, \cdots, K)$ are independently and
isotropically distributed on the $n_T$--dimensional complex unit
sphere due to the assumption of i.i.d. Rayleigh fading. Thus we can
conclude that the orthonormal basis $\hat{\mathbf{q}}_1, \cdots,
\hat{\mathbf{q}}_K$ of the subspace spanned by quantized channel
vectors $\hat{\mathbf{h}}_k ( k =1, \cdots, K)$ have no preference
of direction, i.e., $\left [\hat{\mathbf{q}}_1^T, \cdots,
\hat{\mathbf{q}}_K^T \right]^T$ is isotropically distributed in the
$K \times n_T$ semi-unitary space. Thus, to derive the distribution
of $\varepsilon_1 = \| \tilde{\mathbf{h}}_1 \hat{\mathbf{Q}}^H
\|^2$, we can assume $\hat{\mathbf{q}}_i= \mathbf{e}_i$ for $i =1,
\cdots, K $ without loss of generality, where $\mathbf{e}_i$ is the
$i$-th row of the identity matrix $\mathbf{I}_{n_T}$. Recall that
$\tilde{{\mathbf{h}}}_1 \bot \hat{\mathbf{q}}_1 $, thus the random
vector $\tilde{{\mathbf{h}}}_1$ can be written in the form of
$\tilde{{\mathbf{h}}}_1 = [0, \mathbf{v}]$, where the vector
$\mathbf{v} = [v_1, v_2, \cdots, v_{n_T -1}]$ is isotropically
distributed on the $\left(n_T-1\right)$--dimensional complex unit
sphere. Then $\varepsilon_1 = \| \tilde{\mathbf{h}}_1
\hat{\mathbf{Q}}^H \|^2 = \sum_{l=1}^{K-1} |v_l|^2$. Let $t_l =
|v_l|^2$. It has been obtained in \cite{L.sun_multiuser} that the
joint p.d.f. of $t_1,\cdots,t_{K-1}$ is
\begin{eqnarray}\label{eq:joint_density}
f (t_1,\ldots,t_{K-1}) = \left\{ \begin{array}{ll}
\frac{\Gamma(n_T-1)}{\Gamma(n_T-K)} \left(1 - \sum_{i=1}^{K-1}t_i
\right)^{n_T-K-1}, & t_i \geq 0~ \text{for} ~ i= 1,\cdots, K-1,
\sum_{i=1}^{K-1} t_i=1\\
0, &\text{otherwise}
\end{array}
\right.. \nonumber
\end{eqnarray}

Now we want to obtain the distribution of $u_1= \sum_{l=1}^{K-1}
t_l$. We define the following transformation of variables
\begin{align}
u_1 =\sum_{l=1}^{K-1} t_l, ~~u_i=t_i ~~~\text{for}~ i =2, \cdots,
K-1. \nonumber
\end{align}
It is easy to obtain the corresponding Jacobian is $J = 1$. Thus the
joint p.d.f. of $u_1, \cdots, u_{K-1}$ is
\begin{equation}
f_{u_1,\ldots,u_{K-1}}\left( x_1, \cdots, x_{K-1}\right) =
\frac{\Gamma(n_T-1)}{\Gamma(n_T-K)} \left(1 - x_1 \right)^{n_T-K-1}.
\end{equation}
Since $ 0 \leq t_i \leq 1$, we have $0 \leq t_1 = u_1 -
\sum_{i=2}^{K-1}u_i \leq 1$. The region of the random variables
after transformation can be obtained as $\mathcal{D} = \big\{ \left(
u_1, \cdots, u_{K-1} \right) ~|~ 0\leq \sum_{l=2}^{K-1}u_l\leq u_1
\leq 1,  0 \leq u_i \leq 1 ~\text{for} ~i = 2, \cdots, K-1 \big\}$.
Then the marginal distribution of $u_1$ can be obtained as
\begin{eqnarray}
f_{u_1} (x) &&= \int \cdots \int_{\mathcal{D}} f (x,x_2
\ldots,x_{K-1}) ~{\rm d}x_2 \cdots{\rm d}x_{K-1} \nonumber \\
&& = \int \cdots \int_{\mathcal{D}}
\frac{\Gamma(n_T-1)}{\Gamma(n_T-K)} \left(1 - x_1 \right)^{n_T-K-1}
~{\rm d}x_2 \cdots{\rm d}x_{K-1} \nonumber\\
&& \mathop = \limits^{(a)} \frac{\Gamma(n_T-1)}{\Gamma(n_T-K)} (1-
x)^{n_T-K-1} \frac{x^{K-2}}{\left(K-2\right) !} \nonumber
\end{eqnarray}
which is given by (\ref{eq:pdf_interfers}), where in (a) we have
used the identity $\mathop {\int \int {  \cdots \int {} }
}\limits_{\scriptstyle \sum_{i=1}^{n}t_i \leq h \hfill \atop
\scriptstyle t_1 \geq 0, \cdots, t_{n} \geq 0  \hfill} ~{\rm d} t_1
\cdots ~{\rm d} t_n = \frac{h^n}{n!} $ \cite{Gradshteyn2000}. We
find that $\varepsilon_1 = u_1$ follows beta distribution with shape
$(K-1)$ and $(n_T-K)$. In the following we will prove
$\varepsilon_k$s have the same distribution.

Let $\boldsymbol{\pi}$ be an arbitrary and channel-independent
permutation of $(1,2,\cdots,K)$. $\mathbf{P}_{\boldsymbol{\pi}} =
\big[\mathbf{1}_{\boldsymbol{\pi}(1)}, \cdots, \\
\mathbf{1}_{\boldsymbol{\pi}(K)} \big]^{T}$ is the permutation
matrix corresponding to $\boldsymbol{\pi}$ and
$\mathbf{1}_{\boldsymbol{\pi}(i)}$ is the $\boldsymbol{\pi}(i)$-th
column of identity matrix. We denote
$\hat{\mathbf{H}}_{\boldsymbol{\pi}} =
\mathbf{P}_{\boldsymbol{\pi}}\hat{\mathbf{H}} = \left [
\mathbf{h}_{\boldsymbol{\pi}(1)}^T,\cdots,\mathbf{h}_{\boldsymbol{\pi}(K)}^T\right]^T$
the matrix obtained by permutating the row vector of matrix
$\hat{\mathbf{H}}$ according to the permutation $\boldsymbol{\pi}$.
Then the QR decomposition of $\hat{\mathbf{H}}_{\boldsymbol{\pi}}$
can be written as $\hat{\mathbf{H}}_{\boldsymbol{\pi}}
=\mathbf{P}_{\boldsymbol{\pi}} \hat{\mathbf{R}} \hat{\mathbf{Q}}=
\hat{\mathbf{R}}_{\boldsymbol{\pi}}
\hat{\mathbf{Q}}_{\boldsymbol{\pi}}$. With the assumption that
$\hat{\mathbf{R}}_{\boldsymbol{\pi}}$ has positive diagonal
elements, the above QR decomposition of
$\hat{\mathbf{H}}_{\boldsymbol{\pi}}$ is unique. Using Givens
transformation, there is a series of Givens matrices $\mathbf{G}_1,
\cdots, \mathbf{G}_{K-1} \in \mathbb{C}^{K\times K} $ which satisfy
$\mathbf{P}_{\boldsymbol{\pi}} \hat{\mathbf{R}} \mathbf{G}_1 \cdots
\mathbf{G}_{K-1}= \bar{\mathbf{R}}_{\boldsymbol{\pi}}$
\cite{matrix_com}, where $\bar{\mathbf{R}}_{\boldsymbol{\pi}} \in
\mathbb{C}^{K\times K} $ is a lower triangular matrix with positive
diagonal elements. Since Givens matrix is unitary, we have
$\mathbf{G}_1 \cdots \mathbf{G}_{K-1}\mathbf{G}_{K-1}^H \cdots
\mathbf{G}_1^H= \mathbf{I}$. So $\hat{\mathbf{H}}_{\boldsymbol{\pi}}
$ can be written as $\hat{\mathbf{H}}_{\boldsymbol{\pi}} =
\bar{\mathbf{R}}_{\boldsymbol{\pi}} \mathbf{G}_{K-1}^H \cdots
\mathbf{G}_1^H \hat{\mathbf{Q}} $. Let
$\bar{{\mathbf{Q}}}_{\boldsymbol{\pi}} = \mathbf{G}_{K-1}^H \cdots
\mathbf{G}_1^H \hat{\mathbf{Q}} $. Then we have
$\hat{\mathbf{H}}_{\boldsymbol{\pi}}
=\bar{\mathbf{R}}_{\boldsymbol{\pi}}
\bar{\mathbf{Q}}_{\boldsymbol{\pi}} $ where
$\bar{{\mathbf{Q}}}_{\boldsymbol{\pi}} $ is unitary. Thus
$\hat{\mathbf{H}}_{\boldsymbol{\pi}}
=\bar{\mathbf{R}}_{\boldsymbol{\pi}}
\bar{\mathbf{Q}}_{\boldsymbol{\pi}} $ is also a QR decomposition of
$\hat{\mathbf{H}}_{\boldsymbol{\pi}}$. Using the uniqueness of QR
decomposition, we conclude that $\bar{\mathbf{Q}}_{\boldsymbol{\pi}}
= \hat{\mathbf{Q}}_{\boldsymbol{\pi}}$ and
$\bar{\mathbf{R}}_{\boldsymbol{\pi}}
=\hat{\mathbf{R}}_{\boldsymbol{\pi}}$. Thus we have
\begin{align}
\varepsilon_k &= \| \tilde{\mathbf{h}}_k \hat{\mathbf{Q}}^H \|^2
\nonumber\\
& = \| \tilde{\mathbf{h}}_k \hat{\mathbf{Q}}_{\boldsymbol{\pi}}^H
\mathbf{G}_{K-1}^H  \cdots \mathbf{G}_{1}^H\|^2 \nonumber\\
& \mathop = \limits^{(b)} \| \tilde{\mathbf{h}}_k
 \bar{\mathbf{Q}}_{\boldsymbol{\pi}}^H \|^2,
\end{align}
where $(b)$ is due to the fact that the matrix $\mathbf{G}_{i} $ is
unitary for $i =1, \cdots, K-1$. If we let $\boldsymbol{\pi}(1) =
k$, $\hat{\mathbf{h}}_{k}$ will be the first row of
$\hat{\mathbf{H}}_{\boldsymbol{\pi}}$. According to the previous
derivation in the proof, we know $\varepsilon_k$ for $k=2, \cdots,
K-1$ have the same distribution as $\varepsilon_1$ whose p.d.f. is
give by (\ref{eq:pdf_interfers}).

\section{Proof of \emph{Lemma}
\ref{lemma:exp_interference}}\label{sec:proof_exp_interference}

Let $Y =  \| \tilde{\mathbf{h}}_k \hat{\mathbf{Q}}^H \|^2$. As shown
in \emph{Lemma} \ref{lemma:beta_k_pdf}, $Y$ follows the beta
distribution with shape $(K-1)$ and $(n_T-K)$ and the cumulative
distribution function (c.d.f.) is given by $\text{Pr}\left( Y \leq y
\right) =  I_{x}\left(K-1, n_T - K \right)$, where $I_{x}\left(
\cdot, \cdot \right)$ is the regularized incomplete beta function.
Using the facts that
\begin{align}
I_{x}\left( a, b \right) = \sum_{m = a}^{a+b-1}
\frac{(a+b-1)!}{m!(a+b-m-1)!} x^{m}(1-x)^{a+b-m-1},
\end{align}
$\mathbb{E}[X] = \int_{0}^{\infty} \text{Pr}\left( X \geq x \right)
{\rm d}x $ for nonnegative random variables and binomial expansion,
we have
\begin{align}
\mathbb{E} \left [ - \ln Y \right ] & = \int_{0}^{\infty}
\text{Pr}\left( Y \leq e^{-x} \right) {\rm d} x
\nonumber\\
& = \int_{0}^{\infty} I_{e^{-x}} (K-1,n_T -K) {\rm d} x \nonumber\\
& = \int_{0}^{\infty}  \left\{ \sum_{m = K-1}^{n_T-2}
\frac{(n_T-2)!}{m!(n_T-2-m)!} e^{-mx}(1-e^{-x})^{n_T-2-m}
\right\}{\rm d} x  \nonumber\\
& = \int_{0}^{\infty} \left\{ \sum_{m = K-1}^{n_T-2}
\frac{(n_T-2)!}{m!(n_T-2-m)!} e^{-mx} \sum_{l=0}^{n_T -m-2}
\binom{n_T -m-2}{l} (-1)^{l} e^{- l x} \right\} {\rm d} x
\nonumber\\
&=  \sum_{m = K-1}^{n_T-2} \sum_{l=0}^{n_T -m-2}
\frac{(n_T-2)!}{m!(n_T-2-m)!}\binom{n_T -m-2}{l} (-1)^{l}
\frac{1}{m+l} \nonumber\\
&=  \sum_{m = K-1}^{n_T-2} \sum_{l=0}^{n_T -m-2} \frac{(n_T-2)!}{
m!~ l!~ (n_T -m-2 - l) !} (-1)^{l} \frac{1}{m+l}.
\end{align}
Thus (\ref{eq:exp_interference}) is proved.

\section{Proof of \emph{Theorem}
\ref{theorem:rate_loss}}\label{proof:theorem_rateloss} First we will
prove the fact that $|r_{k,k}|^2$ and $\rho_k^2 |\hat{r}_{k,k}|^2$
have the same distribution. Let the QR decomposition of matrix
$\check{\mathbf{H}} = \mathbf{\Phi} \hat{\mathbf{H}}$ be
$\check{\mathbf{Q}}\check{\mathbf{R}}$. It is easy to see
$|\check{r}_{k,k} |^2 = \rho_k^2 |\hat{r}_{k,k}|^2$, where
$\check{r}_{k,k}$ is the $k$-th diagonal element of
$\check{\mathbf{R}}$. Since we assume using RVQ,
$\check{\mathbf{H}}$ has the same distribution as $\mathbf{H}$. Thus
$\rho_k^2 |\hat{r}_{k,k}|^2$ has the same distribution as
$|r_{k,k}|^2$.

Using (\ref{eq:SNR_per}), (\ref{eq:SINR_limited})
(\ref{eq:rate_per}) and (\ref{eq:rate_limit}), we can write
\begin{align}
\Delta R_{k} & = \mathbb{E}_{\mathbf{H}, \mathbb{W}}  \left \{
\log_2\left(1+ \frac{P}{\kappa} |r_{k,k}|^2 \right) - \log_2\left(1
+ \frac{\frac{P}{\kappa} \rho_k^2|\hat{r}_{k,k}|^2 \cos^2\theta_k}{
\frac{P}{\kappa} \rho_k^2 \| \tilde{\mathbf{h}}_k \hat{\mathbf{Q}}^H
\|^2 \sin^2\theta_k+1 }
\right) \right \}\nonumber\\
\label{eq:rate_loss_4} &= \mathbb{E}_{\mathbf{H}, \mathbb{W}}  \left
\{ \log_2\left(1+ \frac{P}{\kappa} |r_{k,k}|^2 \right) \right\} -
\mathbb{E}_{\mathbf{H}, \mathbb{W}}  \left \{ \log_2 \left(
\frac{P}{\kappa}\rho_k^2 \left ( |\hat{r}_{k,k}|^2 \cos^2\theta_k +
\| \tilde{\mathbf{h}}_k \hat{\mathbf{Q}}^H \|^2 \sin^2\theta_k
\right ) +1 \right) \right \} \nonumber\\
& \hspace{2cm} + \mathbb{E}_{\mathbf{H}, \mathbb{W}}  \left \{
\log_2 \left( \frac{P}{\kappa} \rho_k^2 \| \tilde{\mathbf{h}}_k
\hat{\mathbf{Q}}^H \|^2
\sin^2\theta_k+1 \right)\right\} \\
\label{eq:rate_loss_1} & \leq \mathbb{E}_{\mathbf{H}, \mathbb{W}}
\left \{ \log_2\left(1+ \frac{P}{\kappa} |r_{k,k}|^2 \right)
\right\} - \mathbb{E}_{\mathbf{H}, \mathbb{W}}  \left \{ \log_2
\left( 1+ \frac{P}{\kappa}\rho_k^2 |\hat{r}_{k,k}|^2 \cos^2\theta_k
\right)
\right \}\nonumber\\
 & \hspace{2cm} + \mathbb{E}_{\mathbf{H}, \mathbb{W}}  \left \{  \log_2
\left(1+ \frac{P}{\kappa} \rho_k^2 \| \tilde{\mathbf{h}}_k
\hat{\mathbf{Q}}^H \|^2
\sin^2\theta_k \right)\right\} \\
\label{eq:rate_loss_2} &\approx - \mathbb{E}_{\mathbf{H},
\mathbb{W}}  \left \{ \log_2\left( \cos^2\theta_k \right)\right \} +
\mathbb{E}_{\mathbf{H}, \mathbb{W}}  \left \{  \log_2 \left( 1+
\frac{P}{\kappa} \rho_k^2 \| \tilde{\mathbf{h}}_k \hat{\mathbf{Q}}^H
\|^2 \sin^2\theta_k \right)\right\} \\
\label{eq:rate_loss_3} & \leq \log_2 \left( 1+ \frac{P}{\kappa}
\mathbb{E}_{\mathbf{H}, \mathbb{W}}  \left \{ \rho_k^2 \|
\tilde{\mathbf{h}}_k \hat{\mathbf{Q}}^H \|^2 \sin^2\theta_k \right\}
\right) - \mathbb{E}_{\mathbf{H}, \mathbb{W}} \left \{ \log_2\left(
\cos^2\theta_k \right)\right \}.
\end{align}
Here (\ref{eq:rate_loss_1}) holds by eliminating the non-negative
terms $\| \tilde{\mathbf{h}}_k \hat{\mathbf{Q}}^H \|^2
\sin^2\theta_k$ in the second term of (\ref{eq:rate_loss_4}).
(\ref{eq:rate_loss_2}) follows by using high SNR approximation and
the fact that $|r_{k,k}|^2$ and $\rho_k^2 |\hat{r}_{k,k}|^2$ have
the same distribution which has been proved above.
(\ref{eq:rate_loss_3}) follows by applying Jensen's inequality.

Since the norm of the channel vector $\rho_k$ and the direction of
channel vector $\bar{\mathbf{h}}_k$ are independent and
$\sin^2\theta_k$ and $\tilde{\mathbf{h}}_k$ ($\hat{\mathbf{h}}_k$)
are also independent with each other \cite{Jindal06}, we have
\begin{align}\label{eq:expectation}
\mathbb{E}_{\mathbf{H}, \mathbb{W}}  \left \{ \rho_k^2 \|
\tilde{\mathbf{h}}_k \hat{\mathbf{Q}}^H \|^2 \sin^2\theta_k
\right\}& = \mathbb{E}_{\mathbf{H}, \mathbb{W}} \left \{ \rho_k^2
\right\} \mathbb{E}_{\mathbf{H}, \mathbb{W}} \left \{\|
\tilde{\mathbf{h}}_k \hat{\mathbf{Q}}^H \|^2  \right\}
\mathbb{E}_{\mathbf{H}, \mathbb{W}} \left \{ \sin^2\theta_k
\right\}.
\end{align}
Each term of right hand side of (\ref{eq:expectation}) can be
obtained respectively as follows.
\begin{equation}\label{eq:expectation1}
\mathbb{E}_{\mathbf{H}, \mathbb{W}} \left \{ \rho_k^2 \right\}
=\mathbb{E}_{\mathbf{H}, \mathbb{W}} \left \{ \chi^2_{2 n_T}
\right\}  =  n_T,
\end{equation}
\begin{equation}\label{eq:expectation2}
\mathbb{E}_{\mathbf{H}, \mathbb{W}} \left \{\| \tilde{\mathbf{h}}_k
\hat{\mathbf{Q}}^H \|^2 \right\} = \mathbb{E}_{\mathbf{H},
\mathbb{W}}  \left( \text{Beta}\left(K-1,n_T-1\right) \right) =
\frac{K-1}{n_T-1},
\end{equation}
\begin{equation}
\label{eq:expectation3} \mathbb{E}_{\mathbf{H}, \mathbb{W}} \left
\{\sin^2\theta_k \right\} \leq  2^{- \frac{B}{M-1}},
\end{equation}
where (\ref{eq:expectation2}) can be easily obtained by using p.d.f.
result in (\ref{eq:pdf_interfers}) and (\ref{eq:expectation3}) is
given in \cite{Au-Yeung07} and\cite[\emph{Lemma} 1]{Jindal06}
respectively. In \cite{Bhagavatula11_adaptive}, the second term in
(\ref{eq:rate_loss_3}) was obtained as
\begin{align}\label{eq:rate_loss_2ndterm0}
\mathbb{E}_{\mathbf{H}, \mathbb{W}}  \left \{ \log_2\left(
\cos^2\theta_k \right)\right \} = \log_2(e) \sum_{i=1}^{n}
\binom{n}{i}(-1)^{i}\sum_{l=1}^{i\left( n_T -1 \right)} \frac{1}{l},
\end{align}
and (\ref{eq:rate_loss_2ndterm0}) is rewritten in
\cite{Bhagavatula11_bit} as
\begin{align}\label{eq:rate_loss_2ndterm}
\mathbb{E}_{\mathbf{H}, \mathbb{W}}  \left \{ \log_2\left(
\cos^2\theta_k \right)\right \} = - \frac{\log_2(e)}{n_T-1}
\sum_{i=1}^{n_T-1} \beta\left( n, \frac{i}{n_T-1}\right).
\end{align}
The final result follows by combing
(\ref{eq:expectation})--(\ref{eq:rate_loss_2ndterm}) .

%

\section{Proof of \emph{Theorem}
\ref{theorem:sum_rate}}\label{proof:theorem_sum_rate}

The average sum rate for user $k$ can be upper bounded as
\begin{align}
\mathbb{E}_{\mathbf{H}, \mathbb{W}} \left\{ R_{Q,k} \right\} &
=\mathbb{E}_{\mathbf{H}, \mathbb{W}}  \left\{ \log_2(1 +
\gamma_k)\right\} \nonumber \\
& =\mathbb{E}_{\mathbf{H}, \mathbb{W}}  \left\{ \log_2 \left(1 +
\frac{\frac{P}{\kappa} \rho_k^2|\hat{r}_{k,k}|^2 \cos^2\theta_k}{
\frac{P}{\kappa} \rho_k^2 \| \tilde{\mathbf{h}}_k \hat{\mathbf{Q}}^H
\|^2 \sin^2\theta_k+1 } \right)\right\} \nonumber \\
& \leq \mathbb{E}_{\mathbf{H}, \mathbb{W}}  \left\{ \log_2 \left(1 +
\frac{|\hat{r}_{k,k}|^2 \cos^2\theta_k}{ \| \tilde{\mathbf{h}}_k
\hat{\mathbf{Q}}^H
\|^2 \sin^2\theta_k } \right)\right\} \nonumber \\
& = \mathbb{E}_{\mathbf{H}, \mathbb{W}}  \left\{ \log_2 \left(
\frac{|\hat{r}_{k,k}|^2 \cos^2\theta_k + \| \tilde{\mathbf{h}}_k
\hat{\mathbf{Q}}^H \|^2 \sin^2\theta_k }{ \| \tilde{\mathbf{h}}_k
\hat{\mathbf{Q}}^H
\|^2 \sin^2\theta_k } \right)\right\} \nonumber \\
\label{eq:ave_sumrate1} & \leq  - \mathbb{E}_{\mathbf{H},
\mathbb{W}}  \left\{ \log_2 \left(
 \| \tilde{\mathbf{h}}_k
\hat{\mathbf{Q}}^H \|^2 \sin^2\theta_k \right)\right\},
\end{align}
where (\ref{eq:ave_sumrate1}) is obtained using the facts that
$|\hat{r}_{k,k}|^2 \leq 1 $ and $ \| \tilde{\mathbf{h}}_k
\hat{\mathbf{Q}}^H \|^2 \leq 1$. By using \cite[\emph{Lemma}
3]{Jindal06} and (\ref{eq:expectation2}), $\mathbb{E}_{\mathbf{H},
\mathbb{W}} \left\{ R_{Q,k} \right\} $ can be upper bounded as
\begin{align}
\mathbb{E}_{\mathbf{H}, \mathbb{W}} \left\{ R_{Q,k} \right\} & \leq
- \mathbb{E}_{\mathbf{H}, \mathbb{W}}  \left\{ \log_2 \left(
 \| \tilde{\mathbf{h}}_k
\hat{\mathbf{Q}}^H \|^2 \right)\right\} - \mathbb{E}_{\mathbf{H},
\mathbb{W}}  \left\{ \log_2 \left( \sin^2\theta_k \right)\right\}.
\end{align}
Then (\ref{eq:ave_sumrate_bound1}) follows by using \emph{Lemma}
\ref{lemma:exp_interference} and \cite[\emph{Lemma} 3]{Jindal06}.

\section{The Proof That the RHS of (\ref{eq:sum_rate_loss}) Can Be
Bounded by a Decreasing Function of $n$ for a Fixed
$n_T$}\label{proof}

Let $\mathcal{J} : = \frac{\log_2(e)}{n_T-1} \sum_{i=1}^{n_T-1}
\beta\left( n, \frac{i}{n_T-1}\right) $. $\beta\left( n,
\frac{i}{n_T-1}\right)$ can be written as
\begin{align}
\beta\left( n, \frac{i}{n_T-1}\right) = \frac{\Gamma(n)
\Gamma(\frac{i}{n_T-1}) }{\Gamma\left( n+ \frac{i}{n_T-1}\right)}.
\end{align}
By applying Kershaw's inequality for the gamma function
\cite{Kershaw},
\begin{align}
\frac{\Gamma(x+ s) }{\Gamma\left(x+1\right)} < \left(x + \frac{s}{2}
\right)^{s-1} , \forall x> 0, 0 \leq s \leq 1.
\end{align}
With $x = n - 1 + \frac{i}{n_T-1}$ and $s = 1 -  \frac{i}{n_T-1}$,
we have
\begin{align}
\frac{\Gamma(n)}{\Gamma\left( n+ \frac{i}{n_T-1}\right)} & \leq
\left( n - \frac{1}{2} + \frac{i}{2(n_T-1)} \right)^{-
\frac{i}{n_T-1}}\leq \left(
 n - \frac{1}{2} \right)^{- \frac{i}{n_T-1}}.
\end{align}
Thus, $\mathcal{J}$ can be upper bounded as
\begin{align}
\mathcal{J} \leq \sum_{i=1}^{n_T-1} \Gamma(\frac{i}{n_T-1}) \left(
 n - \frac{1}{2} \right)^{- \frac{i}{n_T-1}},
\end{align}
where the RHS is a decreasing function of $n$ for a fixed $n_T$.

\bibliographystyle{IEEEtran}
\bibliography{sunliang_bib}
%
%
%
%
%
%
%

\end{document}